\documentclass[letterpaper,11pt,oneside]{article}

\usepackage{amsmath,amssymb}
\usepackage{fancybox}
\newcommand {\ignore} [1] {}

\usepackage{epsfig}
\usepackage{wrapfig}

\normalbaselines
\newtheorem{theorem}{Theorem}

\newtheorem{fact}{Fact}

\newenvironment{proof}{\ni {\em Proof:}}{\hfill $\Box$ \bigbreak}

\newcommand{\junk}[1]{}
\junk{

}

\newcommand{\R}{\ensuremath{\mathbb{R}}} % The reals
\newcommand{\ve}{\varepsilon}

\def \ni {\noindent}

\DeclareMathOperator{\argmin}{argmin}

\begin{document}
\setlength{\parskip}{1.5ex plus0.5ex minus0.5ex}
\setlength{\parindent}{0em}

\title{Improved Hardness of Approximation for Stackelberg Shortest-Path Pricing}

\author{Patrick Briest \thanks{Department of Computer Science, University of Paderborn, Germany. {\tt
     patrick.briest@upb.de}. Work done while the author was staying at Cornell University, supported by a scholarship of the German Academic Exchange Service (DAAD).}\and Sanjeev Khanna
 \thanks{Department of Computer and Information Science, University of Pennsylvania,
   USA. {\tt sanjeev@cis.upenn.edu}.}}

\date{}

\maketitle

\begin{abstract}
We consider the {\em Stackelberg shortest-path pricing problem}, which is defined as follows. Given a graph $G$ with fixed-cost and pricable edges and two distinct vertices $s$ and $t$, we may assign prices to the pricable edges. Based on the predefined fixed costs and our prices, a customer purchases a cheapest $s$-$t$-path in $G$ and we receive payment equal to the sum of prices of pricable edges belonging to the path. Our goal is to find prices maximizing the payment received from the customer. While Stackelberg shortest-path pricing was known to be APX-hard before, we provide the first explicit approximation threshold and prove hardness of approximation within $2-o(1)$.
\end{abstract}

\section{Introduction}
\label{introduction}
The notion of {\em algorithmic pricing} encompasses a wide range of optimization problems aiming to assign revenue-maximizing prices to some fixed set of items given information about the valuation functions of potential customers \cite{Aggarwal04,Guruswami05}. In a line of recent work the approximation complexity of this kind of problem has received considerable attention.

Without supply constraints, the very simple {\em single-price algorithm}, which reduces the search to the one-dimensional subspace of pricings assigning identical prices to all the items, achieves an approximation guarantee of $\mathcal{O}(\log n + \log m)$, where $n$ and $m$ denote the number of item types and customers, respectively \cite{Balcan08,BHK08}. Corresponding hardness of approximation results of $\Omega (\log ^{\ve} m)$ for some $\ve >0$ are known to hold (under different complexity theoretic assumptions) even in the special cases that valuation functions are {\em single-minded} (items are {\em strict complements}) \cite{Demaine06+} or {\em unit-demand} (items are {\em strict substitutes}) \cite{B08,BK07,Chuzhoy+07}. In these cases, it is the potentially conflicting nature of different customers' valuations that constitutes the combinatorial difficulty of multi-dimensional pricing.

Another line of research has been considering so-called {\em Stackelberg pricing} problems \cite{Stackelberg34}, in which valuation functions are expressed implicitly in terms of some optimization problem. More formally, we are given a set of items, each of which has some fixed cost associated with it. In addition to these fixed costs, we may assign prices to a subset of the items. Given both fixed costs and prices, a single customer will purchase a min-cost subset of items subject to some feasibility constraints and we receive payment equal to the prices assigned to items purchased by the customer. As an example, we may think of items as being the edges of a graph and a customer aiming to buy a min-cost spanning tree, cheapest path, etc.

Clearly, as there is only a single customer in this type of problem, conflicting valuation functions can no longer pose a barrier for the design of efficient pricing algorithms. Yet, many Stackelberg pricing problems - and in particular the aforementioned spanning tree and shortest path versions - have so far resisted all attempts at improving over the single-price algorithm's logarithmic approximation guarantee. However, the best known hardness results to date only prove APX-hardness of both the spanning tree \cite{Cardinal07} and shortest path \cite{Joret08} cases without deriving explicit constants.

In this paper, we present the first explicit hardness of approximation result for the shortest path version of Stackelberg pricing, which we show to be hard to approximate within a factor of $2-o(1)$. The result is based on a novel analysis of reduction that is quite similar to the ones previously described in \cite{Joret08} and \cite{Roch05}.

\subsection{Preliminaries}
\label{preliminaries}
In the {\em Stackelberg shortest-path pricing problem} ({\sc StackSP}), we are given a directed graph $G=(V,A)$, a cost function $c:A\to \R_0 ^+$, a distinguished set of {\em pricable edges} $\mathcal{P}\subset A$, $|\mathcal{P}|=m$, and two distinguished nodes $s,t\in V$. We may assign prices $p:\mathcal{P}\to \R _0^+$ to the pricable edges. Given these prices, a consumer will purchase a shortest directed $s$-$t$-path $P^*$ in $G$, i.e.,\[
P^* =\argmin \left\{ \sum _{e\in P}(c(e)+p(e))\, |\, P \mbox{ is $s$-$t$-path} \right\},\]
and we receive revenue $rev(p)=\sum _{e\in P^*}p(e)$. We want to find a price assignment $p$ maximizing $rev(p)$.

Throughout the rest of this paper, we will w.l.o.g. only consider {\sc StackSP} instances for which $c(e)=0$ for all $e\in \mathcal{P}$, i.e., every edge is either pricable or fixed-cost, but never both.

\section{Hardness of Approximation}
\label{hardness}
We are going to show that {\sc StackSP} is quasi-NP-hard to approximate within a factor of $2-o(1)$. The result is obtained by a refined analysis of a construction very similar to the one used previously in \cite{Joret08} and \cite{Roch05}.

\begin{theorem}
\label{t:pathPricing}
{\sc StackSP} cannot be approximated in polynomial time within a factor of $2-2^{-\Omega(\log ^{1-\ve}m)}$ for any $\ve >0$, unless NP $\subseteq$ DTIME($n^{\mathcal{O}(\log n)}$).
\end{theorem}

\subsection{Proof of Theorem \ref{t:pathPricing}}
\label{proof}
The proof of the Theorem is based on a reduction from the {\em label cover problem} ({\sc LabelCover}), which is defined as follows. Given a bipartite graph $G=(V,W,E)$, a set $L=\{ 1,\ldots ,k\}$ of {\em labels} and a set $R_{(v,w)}\subseteq L\times L$ of satisfying label combinations for every edge $(v,w)\in E$, we want to find a label assignment $\ell :V \cup W\to L$ to the vertices of $G$ satisfying the maximum possible number of edges, i.e., edges $(v,w)$ with $(\ell (v),\ell (w))\in R_{(v,w)}$. The following hardness result for {\sc LabelCover}, which is an easy
consequence of the PCP theorem \cite{Arora+98} combined with Raz' parallel repetition theorem \cite{Raz98}, is found, e.g., in the survey by Arora and Lund \cite{Arora+96}.

\begin{theorem}
\label{t:labelCover}
For {\sc LabelCover} on graphs with $n$ vertices, $m$ edges and label set of size $k=\mathcal{O}(n)$ there exists no polynomial time algorithm to decide whether the maximum number of satisfiable edges is $m$ or at most $m/2^{\log ^{1-\ve}m}$ for any $\ve >0$, unless NP $\subseteq$ DTIME($n^{\mathcal{O}(\log n)}$).
\end{theorem}

{\bf Reduction:} Let an instance $G=(V,W,E)$ with label set $L=\{ 1,\ldots ,k\}$ as in Theorem \ref{t:labelCover} be given. Denote $E=\{(v_1,w_1)\ldots ,(v_m,w_m)\}$, where the ordering of the edges is chosen arbitrarily. Note that in our notation $v_i$, $v_j$ for $i\not= j$ may well refer to the same vertex (and the same is true for $w_i$, $w_j$). For ease of notation we denote by $R_i$ the satisfying label combinations for edge $(v_i,w_i)$.

\begin{wrapfigure}{r}{.3\textwidth}
\vspace{-5mm}
  \centering
      \epsfig{file=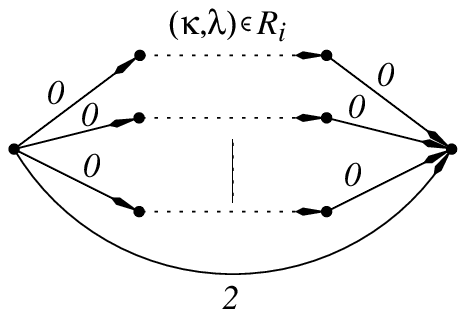,width=50mm}
  \caption{Gadget for an edge $(v_i,w_i)$ in the label cover instance. Each pricable edge corresponds to one satisfying label assignment $(\kappa ,\lambda )$ to vertices $v_i$, $w_i$.}
  \label{fig:gadget}
\vspace{-5mm}
\end{wrapfigure}

We create a {\sc StackSP} instance as follows. For every edge $(v_i,w_i)$ we construct a gadget as depicted in Fig. \ref{fig:gadget}. Essentially, the gadget consist of a set of parallel pricable edges, one for each satisfying label assignment $(\kappa ,\lambda )\in R_i$ and an additional parallel fixed-cost edge of price $2$.

These gadgets are joined together sequentially (see Fig. \ref{fig:reduction}). Let $i<j$ and consider two pricable edges corresponding to label assignments $(\kappa ,\lambda )\in R_i$ and $(\mu ,\nu )\in R_j$. We connect the endpoint of the first edge with the start point of the second edge with a {\em shortcut edge} of cost $j-i-1$, if the two label assignments are conflicting, i.e., if either $v_i=v_j$ and $\kappa \not= \mu$ or $w_i=w_j$ and $\lambda \not= \nu$. This construction is depicted in Fig. \ref{fig:reduction}. Finally, we define the first node in the gadget corresponding to edge $(v_1,w_1)$ and the last node in the gadget corresponding to $(v_m,w_m)$ as nodes $s$ and $t$ the consumer seeks to connect via a directed shortest path. We will refer to the gadgets by their indices $1,\ldots ,m$ and denote the pricable edge corresponding to label assignment $(\kappa ,\lambda)$ in gadget $i$ as $e_{i,\kappa ,\lambda }$.

{\bf Completeness:} Let $\ell$ be a label assignment satisfying all edges in $G$. We define a corresponding pricing $p$ by setting for every pricable edge $p(e_{i,\kappa ,\lambda })=2$ if $\ell (v_i)=\kappa$, $\ell (w_i)=\lambda$ and $p(e_{i,\kappa ,\lambda })=+\infty$ else.

The resulting shortest path from $s$ to $t$ cannot use any of the shortcut edges, because, as $\ell$ is a feasible label assignment, out of any two pricable edges corresponding to conflicting assignments, one must be priced at $+\infty$. Consequently, no path using a shortcut edge can have finite cost. On the other hand, since $\ell$ satisfies every edge, there is a pricable edge of cost $2$ in each of the gadgets. It is then w.l.o.g. to assume that the consumer purchases the shortest path using the maximum possible number of pricable edges and, hence, total revenue is $2m$.

{\bf Soundness:} Let $p$ be a given pricing resulting in overall revenue $m+c$ and let $P$ denote the shortest path purchased by the consumer given these prices. We will argue that there exists a label assignment $\ell$ satisfying $c/4$ of the edges in $G$.

First note that w.l.o.g. any pricable edge that is not part of path $P$ has price $+\infty$ under price assignment $p$. In particular, this means that in every gadget $i$ there is at most a single pricable edge with a finite price. We call this edge the {\em $P$-edge} of gadget $i$. We proceed by grouping gadgets into so-called {\em islands} as detailed below.

{\bf Islands:} Let $\sigma _1$ be the first gadget with a $P$-edge and call $\sigma _1$ the {\em start point} of an {\em island}. Now for each $\sigma _i$ find the maximum value of $j>\sigma _i$, such that gadget $j$ has a $P$-edge and there exists a shortcut edge between the $P$-edges of gadgets $\sigma _i$ and $j$. If such a $j$ exists, define $\sigma _{i+1}=j$, else call $\sigma _i$ an {\em end point}, let $k>\sigma _i$ be the minimum value such that gadget $k$ has a $P$-edge, define $\sigma _{i+1}=k$ and call $\sigma _{i+1}$ a start point. If no such $k$ exists, call $\sigma _i$ an end point and stop. Let $\sigma _r$ be the end point of the final island. We call $\sigma _1,\ldots ,\sigma _r$ the {\em significant gadgets}.

Note that by construction every gadget with a $P$-edge is covered by some {\em island}, i.e., the interval defined by some consecutive start and end points.

\begin{figure}
\centering
\epsfig{file=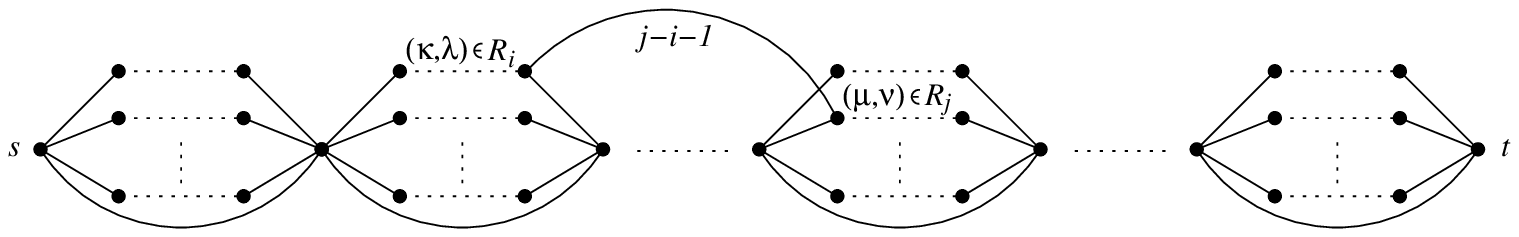,width=15cm}
\caption{\label{fig:reduction} Assembling the edge gadgets into a {\sc StackSP} instance. Conflicting label assignments on two edges $(v_i,w_i)$, $(v_j,w_j)$ are connected by a shortcut of length $j-i-1$. All edges are directed from left to right.}
\end{figure}

\begin{fact}
\label{fact1}
Consider an island $\sigma _{\alpha},\ldots ,\sigma _{\omega}$. Path $P$ does not enter gadget $\sigma _{\alpha}$ or exit gadget $\sigma _{\omega}$ via a shortcut edge.
\end{fact}
\begin{proof}
If $P$ exits $\sigma _{\omega}$ via a shortcut edge, then $\sigma _{\omega}$ could not have been declared an end point. If $\sigma _{\alpha}$ is entered via a shortcut edge, this shortcut must originate from a gadget $i<\sigma _{\alpha}$ which lies within the preceding island. As $P$ cannot bypass the endpoint of the preceding island via a shortcut, $i$ must in fact be the end point $\sigma _{\alpha -1}$ and so $\sigma _{\alpha}$ could not have become a start point.
\end{proof}

Consider now a single island $\sigma _{\alpha},\ldots ,\sigma _{\omega}$. By $\ell _i$ we denote the length of the shortcut edge between gadgets $\sigma _i$ and $\sigma _{i+1}$ for $\alpha \le i\le \omega -1$. Furthermore, by $in_i$ and $out_i$ we refer to the lengths of the shortcut edges used by path $P$ to enter and exit gadget $\sigma _i$, respectively, and set them to $0$ if no shortcuts are used. From Fact \ref{fact1} above it follows that $in_{\alpha}=out_{\omega}=in_{\alpha +1}=0$. See Fig. \ref{fig:island} for an illustration.

For $\alpha \le i\le \omega$ let the cost of path $P$ between shortcut edges $out_i$ and $in_{i+1}$ be $r_i+c_i$, where $r_i$ denotes the cost due to pricable edges and $c_i$ the cost due to fixed-cost edges, respectively. We are going to bound the expression $p_{\sigma _i}+r_i$. We note that $\ell _{\omega}=0$, since by the fact that gadget $\sigma _{\omega}$ is an endpoint, no shortcut edge connects its $P$-edge to the $P$-edge of another gadget. Similarly, we have $r_{\omega}=0$, since path $P$ does not use pricable edges between islands, as we have argued before.

Path $P$ crosses the end node of the $P$-edge in gadget $\sigma _i$ (node $v_2$ in Fig. \ref{fig:island}) and the start node of the $P$-edge of gadget $\sigma _{i+1}$ (node $v_4$ in Fig. \ref{fig:island}) for $\alpha \le i\le \omega -1$. The total cost of path $P$ between these two vertices is $out_i+r_i+c_i+in_{i+1}$. An alternative path $P_1$ is obtained by replacing this part of $P$ with the shortcut edge of length $\ell _i$ between $\sigma _i$ and $\sigma _{i+1}$. By the fact that $P$ is the shortest path we have $out_i+r_i+c_i+in_{i+1}\le \ell _i$ and, thus,
\begin{eqnarray}
\label{eqn1} r_i & \le & \ell _i -out_i-in_{i+1} \quad \quad \mbox{ for } \alpha \le i\le \omega ,
\end{eqnarray}
where the bound on $r_{\omega}$ follows from the fact that for $i=\omega$ all summands in the above expression are $0$. Similarly, the cost of path $P$ between the start node of the shortcut edge into gadget $\sigma _i$ (node $v_1$ in Fig. \ref{fig:island}) and the end node of the shortcut edge exiting $\sigma _i$ (node $v_3$ in Fig. \ref{fig:island}) is $in_i+p_{\sigma _i}+out_i$ for $\alpha \le i\le \omega$. We obtain an alternative path $P_2$ by taking only fixed cost edges of cost $2$ to bypass both shortcuts and gadget $\sigma _i$ at total cost $2(in_i+out_i+1)$. Again, since $P$ is the shortest path, we get $in_i+p_{\sigma _i}+out_i\le 2(in_i+out_i+1)$, or
\begin{eqnarray}
\label{eqn2} p_{\sigma _i} & \le & 2+in_i+out_i \quad \quad \mbox{ for } \alpha \le i\le \omega .
\end{eqnarray}
Combining (\ref{eqn1}) and (\ref{eqn2}) yields
\begin{eqnarray}
\label{eqn3} p_{\sigma _i}+r_i & \le & 2+\ell _i+in_i-in_{i+1} \quad \quad \mbox{ for } \alpha \le i\le \omega .
\end{eqnarray}
Finally, we have
\begin{eqnarray}
\sum _{i=\alpha}^{\omega} \Bigl( p_{\sigma _i}+r_i \Bigr) & \le & \sum _{i=\alpha}^{\omega} \Bigl( 2+\ell _i+in_i-in_{i+1} \Bigr) \\
 & = & 2(\omega -\alpha ) +\sum _{i=\alpha}^{\omega}\ell _i+in_{\alpha}-in_{\omega +1}=2(\omega -\alpha ) +\sum _{i=\alpha}^{\omega}\ell _i.
\end{eqnarray}

Recall that $\sigma _1,\ldots ,\sigma _r$ denote the significant gadgets across all islands. Assume now that there is a total number $I$ of islands with start and end points $\sigma _{\alpha (1)},\sigma _{\omega (1)},\ldots ,\sigma _{\alpha (I)},\sigma _{\omega (I)}$. Summing over all islands we get that overall revenue of price assignment $p$ is bounded by\[
\sum _{j=1}^I\sum _{i=\alpha (j)}^{\omega (j)}p_{\sigma _i}+r_i\le \sum _{j=1}^I\Bigl( 2(\omega (j)-\alpha (j)) +\sum _{i=\alpha (j)}^{\omega (j)}\ell _i\Bigr)\le 2(r-1)+m,\]
where the last inequality follows from the fact that $\alpha (j)=\omega (j-1)+1$ for $2\le j\le I$, $\omega (I)=r$ and $\sum _{i=1}^r\ell _i\le m$, since all shortcuts defining the $\ell _i$ are disjoint. Thus, we have $m+c\le 2(r-1)+m$, or $r\ge c/2+1$.

Now consider the $P$-edges of the $\lceil r/2\rceil$ gadgets $\sigma _1,\sigma _3,\sigma _5,\ldots$ and their corresponding label assignments $(\kappa _i,\lambda _i)$. By definition, there are no shortcut edges between the $P$-edges of any of these gadgets and, thus, $(\kappa _1,\lambda _1),(\kappa _3,\lambda _3),\ldots$ define a non-conflicting label assignment satisfying at least $\lceil r/2\rceil \ge c/4$ edges in $G$. (Labels not defined by $(\kappa _1,\lambda _1),(\kappa _3,\lambda _3),\ldots$ can be chosen arbitrarily.)

Finally, consider a label cover instance as in Theorem \ref{t:labelCover} and the path pricing instance resulting from our reduction above. If all edges can be satisfied, maximum path pricing revenue is $2m$. If no label assignment satisfies more than $m/2^{\log ^{1-\ve}m}$ edges, maximum path pricing revenue is bounded by $(1+4/2^{\log ^{1-\ve}m})m$. This finishes the proof of Theorem \ref{t:pathPricing}.

\begin{figure}
\centering
\epsfig{file=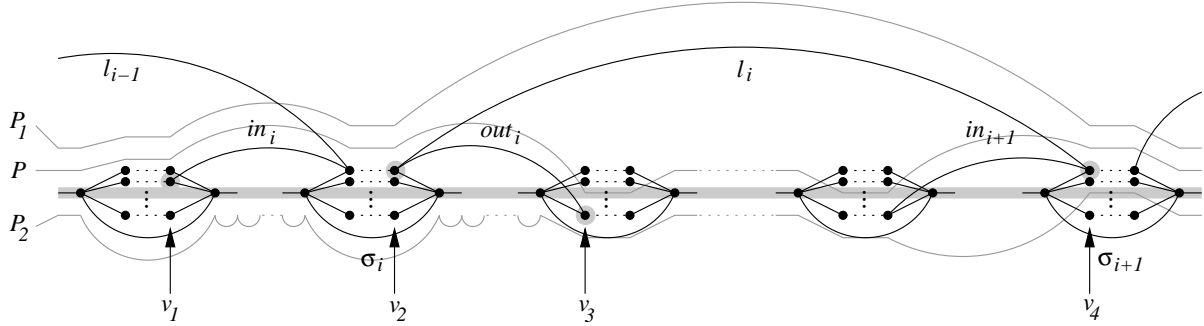,width=16cm}
\caption{\label{fig:island} Two consecutive significant gadgets $\sigma _i$, $\sigma _{i+1}$ inside one island. The length of the shortcut edges used to enter and exit gadget $\sigma _i$ (defined as $0$ if no such shortcut exists) are denoted as $in_i$ and $out_i$, respectively.}
\end{figure}

\section{Tightness}
\label{tightness}
We briefly mention that our analysis is tight in the following sense. It is easy to check that by assigning price $1$ to all pricable edges we can make sure that w.l.o.g. the shortest $s$-$t$-path uses a pricable edge in each of the gadgets and, thus, we obtain revenue $m$. Since maximum possible revenue is bounded above by $2m$ (there is an $s$-$t$-path of that cost that does not use any pricable edges), it follows that it is trivial to achieve approximation guarantee $2$ on the instances resulting from our reduction.

\section{Conclusions}
\label{conclusions}
We have proven the first explicit approximation threshold for any Stackelberg pricing problem. Still, the approximation threshold for this kind of problem in general - and the shortest path version in particular - is far from settled. The following questions seem to constitute fertile ground for future research:
\begin{itemize}
 \item Can we prove super-constant hardness of approximation results for any kind of Stackelberg pricing problem?
 \item Is it possible to achieve a better than logarithmic approximation guarantee for the Stackelberg shortest path pricing problem? Is there an interesting restricted set of graphs on which constant approximation factors are possible?
\end{itemize}

\end{document}